 \documentclass[10pt,twocolumn,twoside]{IEEEtran} 
\IEEEoverridecommandlockouts                              

\usepackage{graphics} 
\usepackage{epsfig} 
\usepackage{mathpazo} 
\usepackage{amsmath} 

\usepackage{amssymb,amsthm}  
\usepackage{amscd}
\usepackage[noadjust]{cite}
\usepackage{float}
\usepackage{times}
\usepackage{url}

\newtheorem{Definition}{Definition}

\newtheorem{Lemma}{Lemma}

\newtheorem{Remark}{Remark}

\newtheorem{theorem}{Theorem}
\newtheorem{pro}[theorem]{Proposition}
\newtheorem{lem}[theorem]{Lemma}
\newtheorem{cor}[theorem]{Corollary}

\newcommand{\R}{\mathbb{R}}
\newcommand{\G}{G}

\renewcommand{\span}{\operatorname{span}}
\renewcommand{\cal}{\mathcal}
\newcommand{\ol}{\overline}
\newcommand{\codim}{\operatorname{codim}}
\newcommand{\GL}{\operatorname{GL}}
\newcommand{\bb}{\mathbb}
\newcommand{\tr}{\operatorname{tr}}

\title{\LARGE \bf
Controllability of Formations over Time-varying Graphs}

\author{Xudong Chen, M.-A. Belabbas, Tamer Ba\c sar
\thanks{Xudong Chen, M.-A. Belabbas, and Tamer Ba\c sar are with Coordinated Science Lab, University of Illinois at Urbana-Champaign,
       Champaign, IL 61820, U.S. Emails: 
       \{xdchen, belabbas, basar1\}@illinois.edu}%
}

\begin{document}

\maketitle
\thispagestyle{empty}
\pagestyle{empty}

\begin{abstract} In this paper, we investigate the controllability of a class of formation control systems. Given a directed graph, we assign an agent to each of its vertices and let the edges of the graph describe the information flow in the system. We relate the strongly connected components of this graph to the reachable set of the formation control system. Moreover, we show that the formation control model is approximately path-controllable over a path-connected, open dense subset as long as the graph is weakly connected and satisfies some mild assumption on the numbers of vertices of the strongly connected components.\end{abstract}

\section{Introduction}

We investigate here the controllability and path-controllability of a non-linear formation control system with $N$ agents in $\R^n$. As is usually done, we use a directed graph $G = (V,E)$, with vertex set $V = \{1,\ldots, N\}$ and edge set  $E$, to describe the information flow in the system. We denote by $i\to j$ an edge in $G$.  Precisely, to each vertex corresponds an agent and  by a slight abuse of notation, we refer to agent $i$ as  $x_i\in \mathbb{R}^n$. Denote by $V^-_i$ the set of out-neighbors of $i$: $
V^-_i:= \{j\in V\mid i\to j\in E\}
$. The  motion of  agent $x_i$  is given by 
\begin{equation}\label{MODEL}
\dot {x}_i = \sum_{j\in V^-_i} u_{ij}(t,x) (x_j - x_i)
\end{equation} 
where each $u_{ij}$ is an integrable real-valued function. 
This formation control model and  variations of it have been widely investigated in recent years \cite{GP,krick2009,XC2014ACC,XC2015ACC,BDOthree,AB2012CDC,sun2014CDC,USZB,mou2014CDC,AL2014ECC,AB2013TAC}. 
Questions about how these scalar functions, $u_{ij}$'s, are designed to organize multi-agent systems~\cite{krick2009,GP}, questions about convergence of the dynamics~\cite{XC2014ACC}, questions about local/global stabilization of the target formation~\cite{AL2014ECC, AB2013TAC}, and questions about robustness of the formation control laws~\cite{AB2012CDC,sun2014CDC,USZB,mou2014CDC} have all been investigated to some extent. 

In this paper, we investigate  whether we can steer the multi-agent system~\eqref{MODEL} from any initial configuration to any target configuration through the choice of the $u_{ij}$'s. The same question was addressed earlier for an undirected graph~\cite{XC2014CDC}. It was shown in~\cite{XC2014CDC} that if the undirected graph $G$ is connected and $(N-n)>1$, $n>1$, then the control system is controllable over a path-connected, open dense subset of the configuration space (comprised of configurations with fixed centroid). We assume here without loss of generality that $G$ is weakly connected and that $N>n$. In case $\G$ is not weakly connected, one can analyze the weakly connected components independently using the results of this paper, and in case $N \leq n$, one can see that the dynamics~\eqref{MODEL} evolves in a proper affine subspace of $\R^n$ with its dimension less than $N$. Thus, one can  use the results of this paper, after a simple change of variables, to study that case as well.

One of main contributions of this paper is to identify a class of weakly connected directed graphs for which the  system~\eqref{MODEL} is controllable. In particular, we will establish a relation between the geometry of formations,  the structure of the underlying network topology, and the controllability of the formation control system. This paper expands on the preliminary version~\cite{ChenBelabbasBasarDirectedCDC2015} by, among others, providing an analysis of the formation control system~\eqref{MODEL} with time-varying graphs,  a finer description of their reachable sets 
and  proofs that were omitted.

Following this introduction, the remainder of the paper is organized as follows. In the next section, we  introduce some  definitions and state the main theorem. We also derive properties of the configuration space of the formation control system; in particular, we  identify an open dense subset of the configuration space where  system~\eqref{MODEL} is controllable.  We obtain a necessary and sufficient condition for this open dense subset to be path-connected. Next, we introduce the matrix Lie algebra $\bb{A}$ of zero row-sum matrices and show how to relate the graph closure of $G$ to the Lie algebraic closure of a naturally defined subspace of $\bb{A}$.  
In section 3, we  compute the Lie brackets of control vector fields and prove the controllability of system~\eqref{MODEL} by verifying the Lie algebra rank condition. We summarize and provide future directions in the last section.

\section{Preliminaries and statement of the main result}

\subsection{Digraphs and their strong component decompositions}\label{sec:prelimdig}
Let $G = (V,E)$ be a directed graph (or simply {\it digraph}) of $N$ vertices with $V = \{v_1,\ldots, v_N\}$ the set of vertices and $E$ the set of edges. We denote by $v_i\to v_j$ a directed edge in $G$  from $v_i$ to $v_j$.  We call a digraph $G$ {\bf weakly connected} if the undirected graph obtained by ignoring the orientation of the edges is connected~\cite{diestel_graph_2010}. The digraph $G$ is {\bf strongly connected} if for any pair of vertices $v_i$ and $v_j$, there is a path in $G$ from $v_i$ to $v_j$.  We say that $G_i = (V_i, E_i)$ is a subgraph of $G$ if $V_i \subset V$ and $E_i \subset E$. We call two subgraphs $G_i$ and $G_j$ {\bf disjoint} if $V_i \cap V_j = \emptyset$. Furthermore, we say that $G_i$ is {\bf induced by} $V_i$ if $E_i$ contains all edges in $E$ that connect vertices in $V_i$, i.e. $$E_i := \{v_k \to v_l\in E \mid v_k,v_l \in V_i\}.$$

\begin{Definition}[Strong component decomposition]\label{def:strongcompdecomp} We say that the subgraphs $G_i = (V_i, E_i)$, $1 \leq i \leq q$, form a {\bf strong component decomposition} of $G$ if 
\begin{enumerate}
\item Each subgraph $G_i$, for $1 \leq i \leq q$, is induced by $V_i$, and the $G_i$'s are   pairwise disjoint.
\item The $V_i$'s partition the vertex set $V$:  $\sqcup_{i=1}^q V_i = V$.

\end{enumerate}
\end{Definition}

We are interested in strong component decompositions with the least possible number of subgraphs. We call such a decomposition {\bf coarse}. The following lemma shows that there is a {\it unique} coarse strong component decomposition of a weakly connected digraph.

\begin{Lemma}\label{lem:DefSCD}
Let $G=(V,E)$ be a weakly connected digraph. There is a {unique} strong component decomposition (SCD) of smallest cardinality.
\end{Lemma}

\begin{proof}
We prove Lemma~\ref{lem:DefSCD} by contradiction. Let $q$ be the minimal number of subsets in a SCD. Suppose that there are two distinct sets of subgraphs, $\{G_1, \ldots, G_q\}$ and $\{G'_1,\ldots, G'_q\}$, and that they both are SCDs of $G$ with  $q$ subsets. Then, after  a relabeling of the subgraphs, we can assume that $V_1 \nsubseteq V'_1 $ and $V'_1 \nsubseteq V_1$. Indeed, if this does not hold, then the two SCDs are identical. We now collect  sets $V_i$'s that intersect $V'_1$; define $S$ as follows: if $i\in S$, then $V_{i} \cap V'_1 \neq \emptyset$. Since the $V_i$'s are disjoint and $V'_1 \nsubseteq V_1$, we need at least two sets to cover $V'_1$. Hence, the cardinality of $S$ is at least two.

Now set $V^* := \sqcup_{i\in S}V_{i}$, and let $G^* = (V^*, E^*)$ be the digraph induced by $V^*$. We show below that $G^*$ is strongly connected. Note that if it is the case, then we can obtain a SCD of $G$ whose cardinality is strictly smaller than $q$: indeed, the subgraphs $G_j$'s, for $j\notin S$, together with $G^*$ form a SCD of $G$. Moreover, the number of the subsets of the SCD is strictly smaller than $q$ since the cardinality of $S$ is at least $2$. We thus derive a contradiction,  and hence, conclude that there is a unique SCD with smallest number of subsets.

We now show that $G^*$ is strongly connected. First, we show that for any $v_k\in V'_1$ and any $v_{l_i}\in V_{i}$, for $i\in S$, there is a path $\gamma_{kl_i}$ from $v_k$ to  $v_{l_i}$ and a path $\gamma_{l_ik}$ from $v_{l_i}$ to  $v_k$.  Pick a vertex $v_{k_{i}}\in V_{i} \cap V'_1$. Such a vertex exists by definition of $S$. Then, there is a path from $v_k$ to $v_{k_i}$ in $G'_1$ and a path from $v_{k_i}$ to $v_{l_i}$ in $G_i$. Using these two paths, we can obtain a path $\gamma_{kl_i}$ from $v_k$ to $v_{l_i}$. Using the same argument, we can also obtain a path $\gamma_{l_ik}$ from $v_{l_i}$ to $v_k$.  But then, since $G'_1$ and the $G_i$'s are strongly connected, we can use the paths $\gamma_{l_ik}$ and $\gamma_{kl_j}$ to obtain a path from $v_{l_i}$ to $v_{l_j}$. Using again the fact that the $G_i$'s are strongly connected, we obtain a path from any vertex in $G_i$ to any vertex in $G_j$ and vice-versa. Thus, we have shown that $G^*$ is strongly connected. 
\end{proof}

The coarse strong component decomposition of a weakly connected digraph $G$ induces an acyclic digraph,  with the vertices representing the components of $G$ and edges representing the flows between the components. Precisely, we have the following definition:


\begin{Definition}[Skeleton digraph]
Let $G$ be a weakly connected digraph, and let $G_1,\ldots, G_q$ form the coarse  
 strong component decomposition of $G$. Define an acyclic digraph  $H=(W,F)$ with $q$ vertices  as follows: there is an edge $w_i\to w_j$  in $H$ if and only if there is an edge $v_i\to v_j$ in $G$ with $v_i$ a vertex in $G_i$  and $v_j$  a vertex in $G_j$.  The acyclic digraph $H$ will be referred as the {\bf skeleton digraph} of $G$.  
\end{Definition}

\begin{figure}[h]
\begin{center}
\includegraphics[scale=.42]{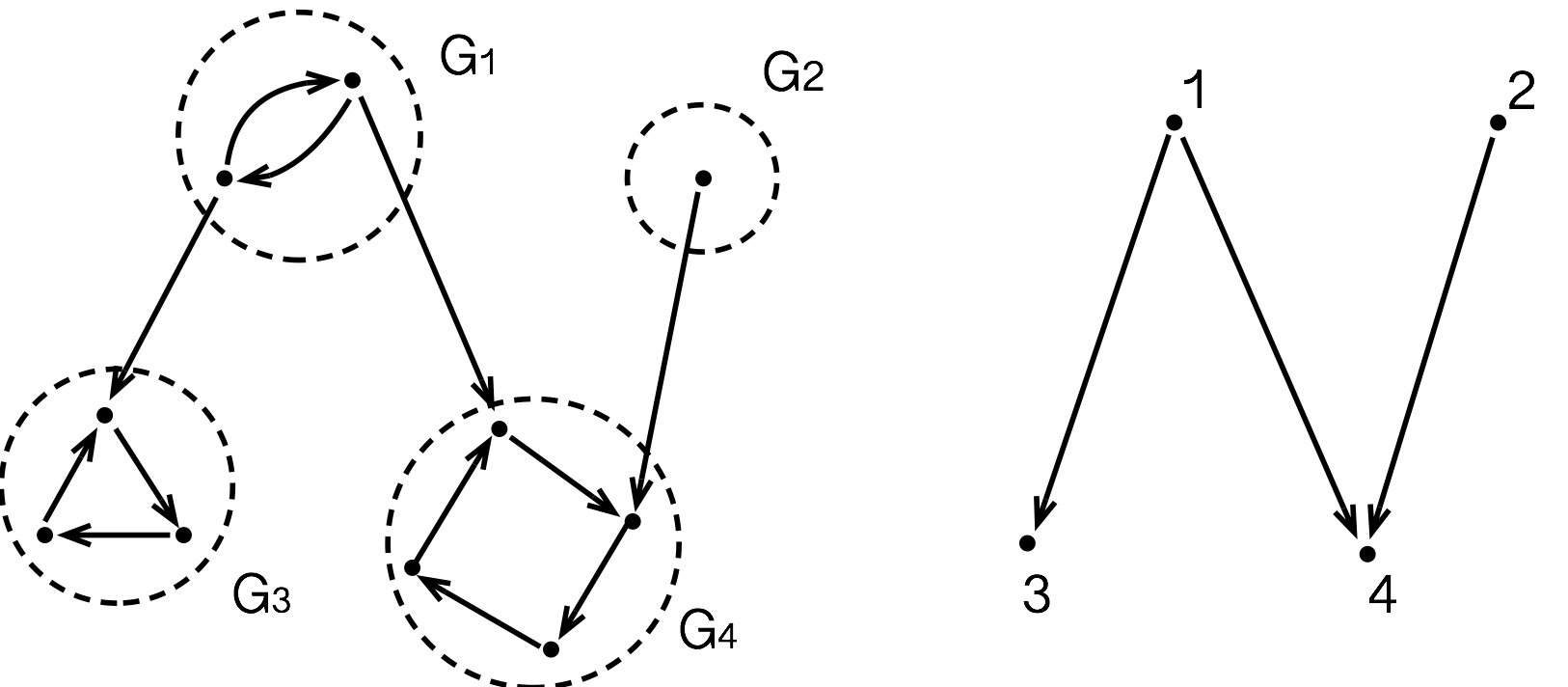}
\caption{By applying the coarse strong component decomposition to the weakly connected graph $G$ on the left, we get four strongly connected subgraphs as $G_1,\ldots,G_4$. The skeleton digraph $H$ of $G$ is  given on the right hand side of the figure. The maximal set $W_+$ of the skeleton of $G$ is given by $\{3,4\}$}
\label{Acyc}
\end{center}
\end{figure}

The digraph $H$ defines a partial order on its vertices: we say $w_j$ is greater than $w_i$, or simply $w_j \succ w_i$, if there is a path from $w_i$ to $w_j$ in $H$. We say a vertex $w_i$ of $H$ is {\bf maximal} for the partial order $\succ$  if there does not exist a vertex $w_j$ such that $w_j\succ w_i$. Denote by $W_+\subseteq W$ 
the set of maximal elements and  refer to it as the {\bf maximal set} of the skeleton of $G$. By definition, if $w_i\in W_+$, then it does not have any outgoing neighbors of $w_i$.  Also, we note that for any $w_i\in W - W_+$,  there is at least a  $w_j\in W_+$, together with a path from $w_i$ to $w_j$. 


\subsection{Configuration space}

Given a formation of $N$ agents in $\R^n$, with states $x_1, \ldots x_N$ respectively, we set $p = (x_1, \ldots, x_N)  \in \R^{nN}$. We call $p$ the {\bf configuration} of the system and $P:=\R^{nN}$ the {\bf configuration space} of the system.  Let $Q$ be a subset of $P$. We say that $Q$ is {\bf path-connected} if for any two configurations $p_0,p_1\in Q$, there is a continuous function $p(t):[0,1]\rightarrow P$ with $p(0)=p_0$ and $p(1)=p_1$ such that the  image of $p(t)$ lies in $Q$. We say that $Q$ is {\bf disconnected} if it is not path-connected.    We now define what it means for a system to be approximately path-controllable:

\begin{Definition}[Approximate path-controllability]
Let $Q$ be a path-connected set.  We say that system~\eqref{MODEL} is {\bf approximately path-controllable} over $Q$ if for any  $T>0$, any arbitrary continuous curve $\hat p:[0,T]\rightarrow Q$ and any tolerance $\epsilon >0$, there are  integrable functions $u_{ij}(p, t)$'s  such that the solution $ p (t)$ of system~\eqref{MODEL}, from any initial condition $p(0)$ with $\|p(0)-\hat p(0)\| < \epsilon$,  satisfies $$\|\hat p(t)- p(t)\|<\epsilon$$ for all $t\in [0,T]$.  
\end{Definition}
We denote by ${\bf u}$ the ensemble of controls $u_{ij}$'s, and let ${\bf u}[0,T]$ be the function ${\bf u}$ over the time interval $[0,T]$.
We now state the main theorem of this paper. 
\begin{theorem}\label{MTHM}
Let $G$ be a weakly connected digraph, and let $G=\{G_1,\ldots,G_q\}$ be the strong component decomposition of $G$. If for each  $w_i\in W_+$ we have  $$|G_i|>(n+1),$$ then system~\eqref{MODEL} is approximately path-controllable over a path-connected, open dense subset of $P$. 
\end{theorem}

The path-controllability of system \eqref{MODEL} is established below by verifying the Lie algebra rank condition of the control vector fields. Precisely, we show that the Lie algebra rank condition is satisfied over a path-connected, open dense subset of $P$ as long as the graph $G$ satisfies the assumption of Theorem~\ref{MTHM}. The same proof technique can actually be used to handle time-varying graphs. Let $G(t) = (V,E(t))$ be a right-continuous, time-varying graph. We call a {\bf switching time} a time $t_i$ such that $\lim_{t \to t_i, t<t_i}G(t) \neq G(t_i)$.  We obtain as a corollary of Theorem~\ref{MTHM} the following result:

\begin{cor}\label{COROLLARY} 
Let $G(t)$ be a right-continuous time-varying graph such that for any finite time interval, $G(t)$ has a finite number of switching times. Suppose that for each $t\ge 0$, the graph $G(t)$ satisfies the assumption of Theorem~\ref{MTHM}. Then, the formation control system \eqref{MODEL} is approximately path-controllable  over a path-connected, open dense subset of $P$.
\end{cor}
\begin{proof}

Let $\hat p:[0,T]\rightarrow Q$ be the path we want the system to follow and let $t_1, \ldots, t_m$ be the switching times of $G(t)$. We construct an admissible ${\bf u}[0,T]$ as follows. Given a graph $G(0)$ and an initial configuration $p(0)$ that satisfies $\| p(0)-\hat p(0)\|<\epsilon$, we know from Theorem~\ref{MTHM} that there exists ${\bf u}_1[0,T]$ such that system $(1)$ approximates $\hat p$ over $[0,T]$. We use this control until the first switching time:  ${\bf u}[0,t_1)={\bf u}_1[0,t_1)$. It  follows that   $\|p(t_1)-\hat p(t_1)\|<\epsilon$. We can thus apply Theorem~\ref{MTHM} but with graph $G(t_1)$ to obtain a control law ${\bf u}_2[t_1,T]$ that steers the system from $p(t_1)$ along a trajectory $p(t)$ such that $\|p(t)-\hat p(t)\|<\epsilon$. As before, we let ${\bf u}[t_1,t_2)={\bf u}_2[t_1,t_2)$. Note that implementing the control ${\bf u}$ over the time interval $[0,t_2)$ yields a trajectory $p(t)$ within $\epsilon$ tolerance of $\hat p$ over that interval. Repeating this procedure a finite number of times yields a control $\bf u$ that approximates $\hat p$ as required.\end{proof}



\subsection{Non-degenerate configurations}\label{sec:nondeg}

Let $p$ be a configuration of $m$ agents $x_1,\ldots, x_m$ in $\mathbb{R}^n$. We say that $r_p$ is the {\bf rank} of $p$, if there is no affine subspace of dimension $(r_p-1)$ that contains $p$. Equivalently, it is the dimension of the linear span of $\{x_i - x_1,\ldots, x_i - x_m\}$ for some (and hence, any) $i = 1,\ldots, m$.

\begin{Definition}
[Non-degenerate configuration]
We say that the configuration $p=(x_1,\ldots,x_m) \in \R^{nm}$, $x_i \in \R^n$,  is {\bf non-degenerate} in {$\R^n$} if there is no \emph{proper affine subspace} of $\R^n$ containing $x_1, \ldots, x_m$. Equivalently, $p$ is non-degenerate if $p$ is of full rank, i.e., the linear span of the vectors $\{x_i - x_1,\ldots, x_i - x_m\}$ is $\R^n$.

\end{Definition}
We note that if $p$ is non-degenerate, then the number of agents has to be strictly greater than $n$. 
If $G=(V,E)$ is a digraph with $N$ vertices, a configuration $p\in P$ can be viewed as an embedding of the graph $G$ in $\R^n$ by assigning vertex $v_i$ to $x_i$. We call the pair $(G,p)$ a {\bf framework}. Let $(G, p)$ be a framework with $G$ weakly connected, and let $G_1, \ldots, G_k$ form the coarse strong component decomposition of $G$.  We denote by $(G_i,p_i)$, with $p_i\in \R^{n |G_i|}$,  the framework obtained from $(G,p)$ by only considering vertices and edges of $G_i$. We refer to $p_i$ the {\bf sub-configuration} associated with $G_i$, and similarly denote by $r_{p_i}$ the rank of $p_i$.  
Let $Q$ be the  subset of $P$ defined as follows:
\begin{equation}\label{PCOD}
Q:=\left \{p\in P \mid r_{p_i} = n, \forall \,w_i\in W_+\right\}
\end{equation}
where we recall that $W_+$ is the maximal set of the skeleton of $G$.  
It should be clear that $Q$ is an open subset of $P$. We now show that $Q$ is also dense in $P$, and even path-connected under some mild assumptions.

\begin{pro}\label{LPCOD}
Under the assumption of Theorem~\ref{MTHM}, the set $Q$, defined in~\eqref{PCOD}, is an open dense, path-connected subset of $P$. 
\end{pro}

We can decompose the set $P$ into disjoint components as follows. Fix an integer $k \geq 0$ and define the set
\begin{equation*}
P^k := \{p\in P\mid r_p = k\}.
\end{equation*}
It is easy to see that each $P^k$ is nonempty, for $k = 0,\ldots,n$. Note that the set $P^0,\ldots,P^n$ are pairwise disjoint, and they form a decomposition of $P$:
\begin{equation*}
P = \sqcup^n_{k=0} P^k.
\end{equation*}
Furthermore, for each $k=1,\ldots, n$,  we have
\begin{equation}\label{eq:diffpkpk}
\overline {P^{k}} - P^{k} = \sqcup^{k-1}_{l=0}P^{l}
\end{equation}
where $\overline{P^{k}}$ is the closure of $P^k$ in $\R^{nN}$. 
For $k=0$, we have $$\overline{P^{0}} = P^{0} = \{ p=(x_1,\ldots,x_n) \in P \mid x_1=\cdots=x_n \}.$$ 
Note that $P^n$ is the set of non-degenerate configurations. Since $P^n$ is an open set in $P$, it is a smooth submanifold of $P$. Also, from~\eqref{eq:diffpkpk}, the set $P^n$ is dense in $P$, i.e., 
\begin{equation}\label{eq:clospn}
\ol{P^n} = P. 
\end{equation}
The union of the other sets $P^0,\ldots,P^{n-1}$ is the set of degenerate configurations. Define
\begin{equation}\label{Eq:dimofPk}
d_k :=  -k^2 + k(N+n-1) + n. 
\end{equation}
We now have the following fact:


\begin{lem}\label{pksm}
Each $P^k$ is a smooth submanifold of $P$, and $\dim P^k = d_k$. 
\end{lem}
  
We refer to the Appendix for a complete proof of Lemma~\ref{pksm}.   
For each $k = 0,\ldots, n$, the codimension of $P^k$ in $P$ is by convention defined to be:
$$
\codim P^k := \dim P - \dim P^k.
$$   
We  now establish the following inequalities: 

\begin{lem}\label{codcom}
Suppose that $N > n$. Then,  $$\codim P^k \ge N - n$$ 
for all $k = 0,\ldots, n-1$. 
\end{lem}

\begin{proof}
We prove the result by evaluating the value of $d_k$. On one-hand, from~\eqref{Eq:dimofPk}, we have that $d_k$ is a quadratic function in the variable $k$, and achieves its maximum at $(N + n -1)/2$. Since $N > n$, we have
\begin{equation*}
 \frac{1}{2}(N + n - 1) \geq n, 
\end{equation*}
and thus $d_k$ is a strictly monotonically increasing function in $k$ for $k = 0,\ldots, n-1$, from which we conclude  that 
$$
d_0< \ldots < d_{n-1}.
$$
On the other hand, we have that
\begin{equation*}
d_{n-1} = nN - N + n.
\end{equation*}
Thus, the codimension of $P^k$ satisfies
\begin{equation*}
\codim P^k \ge \codim P^{n-1}= N - n.
\end{equation*}
for all $k = 0,\ldots, n-1$.
\end{proof}

Let $G_i=(V_i,E_i)$, for $i =1,\ldots,q$, form the coarse strong component decomposition of $G$. Let  
$
V_{-i} :=  V - V_i
$, 
and $G_{-i}$ be the subgraph of $G$ induced by $V_{-i}$. Let  $P_i$ and $P_{-i}$ be the sets of sub-configurations associated with $G_i$ and $G_{-i}$, respectively. Let $N_i$ be the number of vertices of $G_i$. We have that $\dim P_i = n\times N_i$ and $\dim P_{-i}=n\times(N-N_i)$. Similarly, we define $P^k_i$ to be the set of rank-$k$ sub-configurations associated with $G_i$. 
We are now in the position to prove Proposition~\ref{LPCOD}.

\begin{proof}[Proof of Proposition~\ref{LPCOD}]
Recall that $W_+$, defined in section~\ref{sec:prelimdig}, is the maximal set of the skeleton of $G$. We first show that $Q$ is an open dense subset of $P$, and then show that $Q$ is path-connected. 
Following the definition of $Q$ (in~\eqref{PCOD}), we can write 
\begin{equation}\label{exprQ}
Q = \bigcap_{w_i\in W_+} \left (P^n_i \times P_{-i} \right )
\end{equation}
From the assumption of Theorem~\ref{MTHM}, we have $N_i - n >1$ for all $w_i\in W_+$. Then, each $P^n_i$ is a nonempty open set in $P_i$, and from~\eqref{eq:clospn}, we have that $\ol{P^n_i} = P_i$ for all $w_i\in W_+$. So then, each $P^n_i \times P_{-i}$ is an open dense subset of $P$, and so is $Q$. This holds  because a finite intersection of open dense subsets is still open and dense.

We now show that $Q$ is path-connected. For each $w_i\in W_+$, define $\widetilde P_i$ to be the set of degenerate configurations in $P_i$. Then, $\widetilde P_i$ can be expressed as follows: 
\begin{equation*}
\widetilde P_i = \bigsqcup^{n-1}_{k=0} P^k_i. 
\end{equation*}
Since $P_i = \widetilde P_i \sqcup P^n_i$, from~\eqref{exprQ}, we can express $Q$ as follows:
\begin{equation*}
Q = P - \bigcup_{w_i\in W_+}\left (\widetilde P_i \times P_{-i}\right ). 
\end{equation*}   
We now show that  each $\widetilde P_i\times P_{-i}$, for $w_i\in W_+$, is a finite union of smooth submanifolds of $P$ and that the codimension of each submanifold  is strictly greater than $1$.   
First, note that 
$$
\widetilde P_i\times P_{-i} = \bigsqcup^{n-1}_{k=0} \left(P^k_i \times P_{-i} \right). 
$$ 
By Lemma~\ref{pksm}, each $P_i^k$ is a submanifold of $P_i$. Futhermore, by Lemma~\ref{codcom}, the codimension of $P^k_i$ in $P_i$ satisfies 
$$
\codim P^k_i:= \dim P_i - \dim P^k_i \ge N_i - n  
$$
for all $k = 0,\ldots, n-1$.  Using the fact that $P = P_{i}\times P_{-i}$, we have  
$$
\codim (P^k_i\times P_{-i}) = \codim P^k_i, 
$$
and hence, 
$$
\codim (P^k_i\times P_{-i}) \ge (N_i - n) > 1. 
$$
We have thus proved that each $\widetilde P_i\times P_{-i}$, for $w_i\in W_+$, is a finite union of smooth submanifolds of $P$, and the codimension of each submanifold is strictly greater than $1$. 
Since removing from a Euclidean space a finite union of smooth submanifolds of codimensions at least two does not render it disconnected, the result is proved.
\end{proof}

Proposition~\ref{LPCOD} shows that the assumption of Theorem~\ref{MTHM} is sufficient for $Q$ to be a nonempty path-connected, open dense subset.  We show below that it is  also necessary for $Q$. Recall that $P^k_i$ is the set of rank-$k$ sub-configurations associated with $G_i$, and in particular, $P^n_i$ is the set of non-degenerate sub-configurations associated with $G_i$. We first establish the following fact.

\begin{pro}\label{pro:disconnect}
Let $W_+$ be the maximal set of the skeleton of $G$.  Suppose that there exists a $w_i\in W_+$ such that $N_i - n\le 1$. Then, the following two results hold: 
\begin{itemize}
\item[1.] If $N_i - n\le 0$, then the set $Q$ is empty. 
\item[2.] If $N_i - n = 1$, then $P^n_i$ has two connected components. 
\end{itemize}
\end{pro}

\begin{proof}
First, observe that if a configuration is non-degenerate in $\bb{R}^n$,  it contains at least $(n+1)$ agents. Thus,  if $N_i \le n$, then there does not exist a non-degenerate sub-configuration associated with $G_i$, and hence,  $Q$ is empty.

We now assume that $N_i - n = 1$, and prove that $P^n_i$ has two connected components. Without loss of generality, we assume that $p_i$ is formed by agents $x_1,\ldots, x_{n+1}$. For each $p_i\in P^n_i$, define a matrix as follows: 
\begin{equation*}
A_{p_i} := (x_2 -x_1,\ldots, x_{n+1} - x_1)\in \R^{n\times n}.
\end{equation*}
Since $p_i$ is non-degenerate, $A_{p_i}$ is invertible. Let $\GL(n)$ be the set of $n$-by-$n$ invertible matrices, and denote by 
$$
f: P^n_{i} \longrightarrow \GL(n)
$$
the smooth map sending $p_i$ to $A_{p_i}$. Note that  $f$ is surjective and open; indeed,  for a matrix $A\in \GL(n)$, with column vectors $a_1,\ldots, a_n\in\bb{R}^n$, we have
$$
f^{-1}(A) = \{(x_1, x_1 + a_1,\ldots, x_1 + a_n) \mid x_1\in\bb{R}^n \}. 
$$   
Thus, we have 
\begin{equation}\label{eq:pnifinvgln}
P^n_i = f^{-1}\GL(n). 
\end{equation}
It is well known that $\GL(n)$ has two connected components:  the matrices with positive determinant and the ones with negative determinant. Thus, following~\eqref{eq:pnifinvgln}, we conclude that  $P^n_i$ has two connected components: the $p_i$'s with $\det(A_{p_i}) > 0$, and the ones with $\det(A_{p_i}) < 0$. This completes the proof. 
\end{proof}

Following Proposition~\ref{pro:disconnect}, we have the following corollary:
\begin{cor}
Let $W_+$ be the maximal set of the skeleton of $G$, and $Q$ be defined in~\eqref{PCOD}.  Suppose that there exists a $w_i\in W_+$ such that $N_i - n\le 1$. Then, $Q$ is disconnected. 
\end{cor}

\begin{proof}
First, from Proposition~\ref{pro:disconnect}, we know that $P^n_i$ has two connected components, and hence, so does $P^n_i\times P_{-i}$. On the other hand, from~\eqref{exprQ}, we have that
$$
Q = (P^n_i \times P_{-i}) - \bigcup_{j\in W_+ - \{i\}} \left(\widetilde P_j\times P_{-j} \right)
$$
from which it follows that the set $Q$ is disconnected. 
\end{proof}

We conclude this section by applying Propositions~\ref{LPCOD} and~\ref{pro:disconnect} to a special case where the digraph $G$ is strongly connected. In this case, the skeleton digraph $H$ of $G$ is comprised of only one vertex, and hence, $Q$ is the set of non-degenerate configurations in $\R^n$. We have the following fact:

\begin{cor}\label{cor:threecases}
Let $G$ be a strongly connected graph with $N$ vertices. Then, the following three properties hold:
\begin{itemize}
\item[1.] If $N< n + 1$, then  $Q$ is empty.
\item[2.] If $N = n + 1$, then $Q$ is open dense in $P$, and it has two connected components. 
\item[3.] If $N > n + 1$, then $Q$ is an open dense, path-connected subset of $P$. 
\end{itemize}
\end{cor}

The first two properties in the corollary follow from Proposition~\ref{pro:disconnect}, and the last one follows from Proposition~\ref{LPCOD}.

\subsection{Non-degenerate sub-configurations}

If $p$ is a non-degenerate configuration in $\mathbb{R}^n$, then there exists at least {\it one} set of $(n+1)$ agents such that the sub-configuration formed by these $(n+1)$ agents is non-degenerate in $\mathbb{R}^n$. We establish below a tighter lower bound on the number of non-degenerate sub-configurations of $p$. 

To this end, let $p = (x_1,\ldots, x_{n+1})$ be a configuration in $\bb{R}^n$ associated with a digraph $G = (V,E)$ with $(n+1)$ vertices.  For each $i \in V$, denote by $S_i$ the affine subspace of $\R^n$ of lowest dimension that contains the $n$ vectors $x_j$'s, for $j\in V - \{i\}$. Note that if $p$ is non-degenerate, then each $S_i$, for $i \in V$, is a hyperplane in $\R^n$, i.e.,  $\dim S_i = n-1$. Similarly,  for any proper subset $V'\subset V$, we define $S_{V'}$ as the affine subspace of $\R^n$ of lowest dimension that contains vectors $x_i$'s, for all $i\in V - V'$. We now establish the following result which relates $S_{V'}$ to the intersection of the hyperplanes $S_i$'s:


\begin{pro}\label{pro:intersection}
Let $p = (x_1,\ldots, x_{n+1})$ be a non-degenerate configuration in $\R^n$ associated with a digraph $G = (V,E)$ with $(n+1)$ vertices.    Let $V'$ be a proper subset of $V$. Then, 
$$S_{V'} = \cap_{i\in V'} S_i.$$
\end{pro}

\begin{proof}
Without loss of generality, we assume that $V' = \{1,\ldots, k\}$ with $k\le n$.   For each $i \in V'$, denote by  $\span \{x_j - x_{n+1} \mid  j\neq i \}$ the linear subspace of $\R^n$ spanned by the $(n-1)$ vectors $\{x_j - x_{n+1} \mid  j\neq i \}$. So then, 
$$
S_i = x_{n+1} + \span \{x_j - x_{n+1} \mid  j\neq i \}.
$$
Since $p$ is non-degenerate, the $n$ vectors $\{x_i - x_{n+1}\mid 1\le i\le n\}$ are linearly independent. Thus, 
$$
\bigcap_{i \in V'} \span  \{x_j - x_{n+1} \mid  j\neq i \} = \span  \{x_j - x_{n+1} \mid  j \notin V' \}, 
$$ 
and hence, it follows that 
$$
\bigcap_{i\in V'} S_i=  x_{n + 1} +  \span  \{x_j - x_{n+1} \mid  j \notin V' \} = S_{V'}.
$$
\end{proof}
 
We obtain a corollary of Proposition~\ref{pro:intersection} as follows:

\begin{cor}\label{cor:intersection}
Let $p = (x_1,\ldots, x_{n+1})$ be a non-degenerate configuration in $\R^n$ associated with a digraph $G = (V,E)$ with $(n+1)$ vertices. Then, for each $i\in V$, we have $$\cap_{j\neq i} S_i = \{x_i\}.$$
\end{cor}  

\begin{proof}
Let $V': = V - \{i\}$, then $S_{V'} = \{x_i\}$; indeed, the affine subspace in $\R^n$ of lowest dimension that contains $x_i$ is the singleton $\{x_i\}$. Then, following Proposition~\ref{pro:intersection},  we have $\cap_{i\in V'} S_i  = S_{V'} = \{x_i\}$.
\end{proof}

We now establish the following fact. 
 
\begin{pro}\label{NUMN}
Let $p=(x_1,\ldots,x_{n+1})$ be a non-degenerate configuration in $\mathbb{R}^n$. Then, for any vector $x \in \mathbb{R}^n$, there exist $n$ vectors $\{x_{i_1},\ldots,x_{i_n}\}\subset \{x_1,\ldots,x_{n+1}\}$ such that these $n$ vectors together with $x$ form a non-degenerate configuration in $\mathbb{R}^n$.  
\end{pro}

\begin{proof}
We prove the result by contradiction. Assume that there is a vector $x$ in $\mathbb{R}^n$ such that there does not exist a set of $n$ vectors  $\{x_{i_1},\ldots,x_{i_n}\}$ out of $\{ x_1,\ldots,x_{n+1}\}$ such that $x,x_{i_1},\ldots,x_{i_n}$ form a non-degenerate configuration in $\mathbb{R}^n$.

Recall that $S_i$ is the affine subspace  of $\R^n$ of lowest dimension that contains the $n$ vectors $x_j$'s, for $j\neq i$. Since $p$ is a non-degenerate configuration, by Corollary~\ref{cor:intersection}, we have
$$
\bigcap_{j \neq i} S_j = \{x_i\}.
$$
So then, we have 
\begin{equation*}
\bigcap^{n+1}_{i = 1} S_i  = \bigcap^{n+1}_{i = 1}\bigcap_{j \neq i} S_j = \bigcap^{n+1}_{i = 1} \{x_i\} =  \emptyset.
\end{equation*}
On the other hand, each $S_i$ has to contain the vector $x$ because otherwise the $(n+1)$ vectors $x_1,\ldots,x_{i-1},$  $x_{i+1},\ldots,x_{n+1},$ $x$  form a non-degenerate configuration in $\mathbb{R}^n$. Thus,  
\begin{equation*}
x\in \bigcap^{n+1}_{i = 1} S_i = \emptyset
\end{equation*}  
which is a contradiction. This completes the proof.
\end{proof}

We  obtain as a corollary a lower bound on the number of non-degenerate sub-configurations of $p$:

\begin{cor}
Let $p \in \R^{nN}$ be a non-degenerate configuration with $N > n$. Then, there are at least $(N-n)$ sub-configurations of $(n+1)$ agents that are non-degenerate in $\mathbb{R}^n$. 
\end{cor}

\subsection{Lie algebra of zero row-sum matrices}\label{sec:LieAlg1}
\begin{Definition}[Zero row-sum matrices]\label{def:zerorowsummatr}
Denote by $\mathbf{1}$ the vector of $\R^N$ with all entries one. We say that a matrix $A \in \R^{N\times N}$ is a {\bf zero row-sum} matrix if $A\mathbf{1}=0$. We denote by $\bb{A}$ the vector space of such matrices. \end{Definition}
It is easy to verify that the commutator or Lie bracket of two zero row-sum matrices is also a zero row-sum matrix. Hence, the vector space $\bb{A}$ is  a Lie algebra. We derive here some properties of the Lie algebra of zero row-sum  matrices that are needed in the proof of the main Theorem. 

Let $e_1,\ldots,e_N$ be the canonical basis of $\mathbb{R}^N$. Let $A_{ij} \in \bb{A}$ be defined as follows:
\begin{equation*}
A_{ij}:= -e_ie^{\top}_i + e_ie^{\top}_j.
\end{equation*} 
Note that the matrix $A_{ij}$ is the negative of the Laplacian matrix of a digraph with $N$ vertices and  only one edge, namely $v_i\to v_j$. For a digraph $G = (V,E)$, define a set of zero-row sum matrices as follows:
\begin{equation*}
A_{G}:= \left\{A_{ij} \in \R^{N \times N} \mid v_i\to v_j \in E \right\}
\end{equation*}
It is easy to see that  matrices in $A_{G}$ are linearly independent. We denote by  $\bb{A}_{G}$ the vector space spanned by elements in $A_{G}$. Further, we introduce the following definitions:
\begin{Definition}
[Lie algebra closure of a vector space of matrices] Given a vector space $\bb{A} \subset \R^{N\times N}$ of matrices, we denote by $\overline{\bb{A}}$ the  {\bf Lie algebra closure} of ${\bb{A}}$, defined  as the vector space of least dimension in $\R^{N\times N}$ which contains $\bb{A}$ and is closed under the matrix Lie bracket. 
\end{Definition}
 
\begin{Definition}[Transitive closure of a digraph]\label{Def: Tranclos}
Given a digraph $G$, we denote by $\overline{G}$  the {\bf transitive closure} of $G$: $\overline G$ has the same vertex set as $G$ and there is  an  edge $v_i\to v_j$  in $\overline{G}$ if and only if there is a path from $v_i$ to $v_j$ in $G$.

\end{Definition}
  For illustration of the transitive closure of a digraph, we refer to the example  depicted in Figure~\ref{TranC}.

\begin{figure}[h]
\begin{center}
\includegraphics[scale=.42]{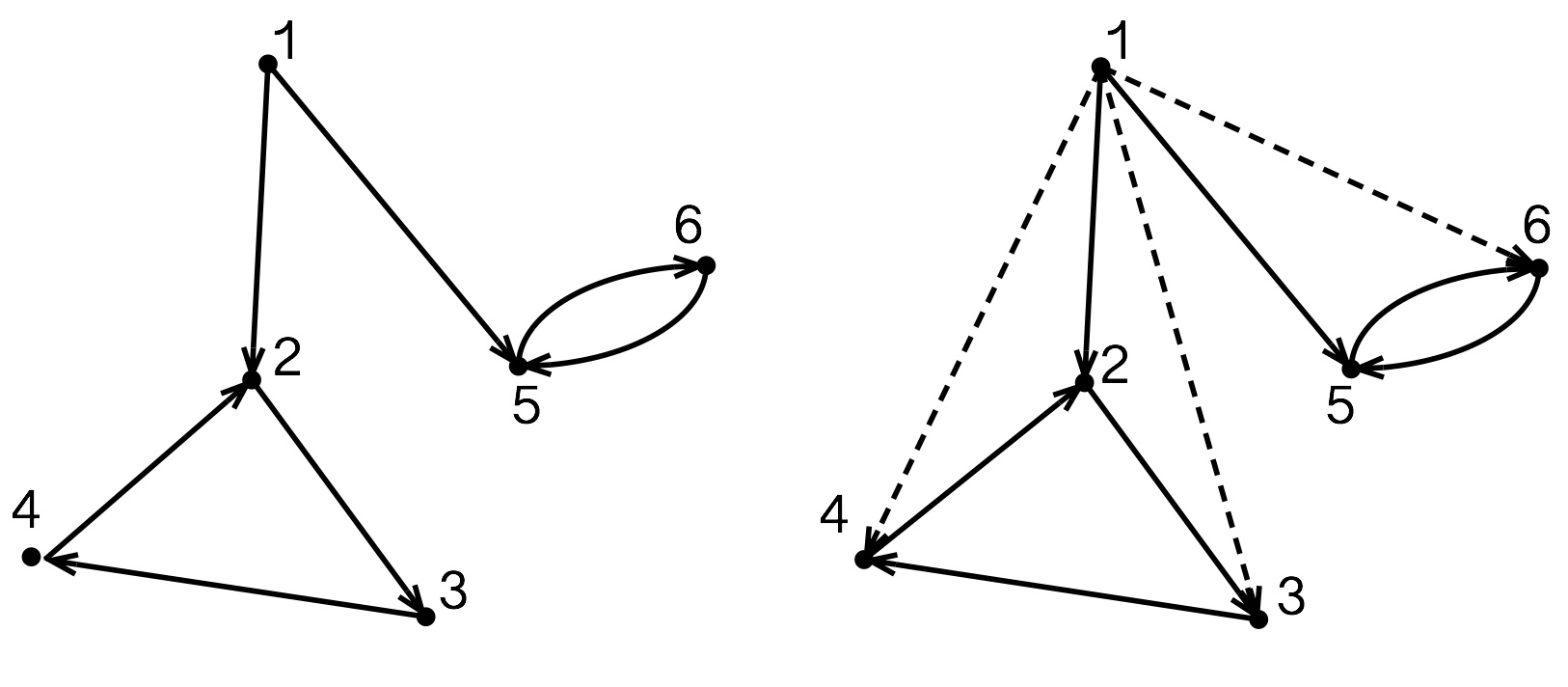}
\caption{The digraph on the right hand side of this figure is the transitive closure of the digraph on the left hand side.}
\label{TranC}
\end{center}
\end{figure}

Our goal in this section is to evaluate $\ol{\bb{A}}_G$ for $G$, a weakly connected graph. In particular, we  establish the following result.

\begin{pro}\label{LIEAL}
Let $G $ be a weakly connected digraph, and  $\overline{G} $ be its transitive closure. Let $\ol{\bb{A}}_G$ be the  Lie algebra closure of $\bb{A}_{G}$.  Then, $$\ol{\bb{A}}_G = \mathbb{A}_{\overline{G}}.$$
\end{pro}

Proposition~\ref{LIEAL} relates the Lie algebra closure  of a set of Laplacian matrices to the transitive closure of a digraph $G$. To prove Proposition~\ref{LIEAL}, 
we  first show that strong component decompositions and transitive closures commute:

\begin{lem}\label{lem:transclos} 
Let $G $ be a weakly connected digraph, and $H $ be the associated skeleton digraph. Then, the following holds:
\begin{itemize}
\item[1.] If $G_1,\ldots, G_q$ form the coarse strong component decomposition of $G$, then $\ol{G}_1,\ldots,\ol{G}_q$ form the coarse strong component decomposition of $\ol{G}$. In particular, each $\ol{G}_i$, for $i=1,\ldots, q$, is a complete graph.
\item[2.] Let $\overline{H}$ be the transitive closure of $H$. Then, $\ol{H}$ is the skeleton digraph of $\overline{G}$. 
\end{itemize}
\end{lem}

We prove the Lemma in the Appendix.

We now evaluate the Lie brackets of matrices in $A_G$.

\begin{lem}\label{COMP}
Let $(i,j)$ and $(i',j')$ be two pairs of positive integers with $1 \leq i\neq j \leq n$ and $1 \leq i'\neq j'\leq n$. Then,  the following three properties hold:
\begin{itemize}
\item[1.] If $i \neq i'$ and $j \neq i'$, then 
$$[A_{ij},A_{i'j'}] =0.$$
\item[2.] If $i=i'$, then
\begin{equation*} 
[A_{ij}, A_{ij'}] = A_{ij} - A_{ij'}.
\end{equation*}
\item[3.] If $j = i'$, then  
\begin{equation*} \label{GENE}
[A_{ij}, A_{jj'}] = A_{ij'} - A_{ij}.
\end{equation*}  
\end{itemize}
\end{lem}

We omit the proof of the above, as the result follows directly from computations. We mention that similar results have been obtained in~\cite{XC2014CDC} and~\cite{CZ2014CDC}.
As a corollary of Lemma~\ref{COMP}, we have the following fact:

\begin{cor}\label{COR} 
If  $v_i\to v_j$ is an edge of $\overline{G}$, then the matrix $A_{ij}$ is contained in $\ol{\bb{A}}_G$.
\end{cor}

\begin{proof}
Let $v_{i}\to v_j$ be an edge of $\ol{G}$; then by the definition of transitive closure, there is a path from $v_i$ to $v_j$ in $G$. Suppose that the path is of length $k$, and we express it as follows:  
$$
v_{i_0}\to \ldots \to v_{i_k}
$$
with $i_0 = i$ and $i_k = j$. We now show that $A_{i_0i_k} \in  \ol{\bb{A}}_G$.  The proof is carried out by induction on the length $k$.
For the base case,  we have $k=2$; then $v_i\to v_j$ is an edge of $G$. Thus, $A_{ij}$ is in $A_{G}$. For the inductive step, suppose that the Lemma holds for $(k-1)$, and we prove for $k$. First, by the induction hypothesis, the matrix $A_{i_0i_{k-1}}$ is in $\ol{\bb{A}}_G$. Then, from  Lemma~\ref{GENE}, we have  
\begin{equation*}
A_{i_0i_{k}} = [A_{i_0i_{k-1}},A_{i_{k-1}i_{k}}] + A_{i_0i_{k-1}} \in \ol{\bb{A}}_G.
\end{equation*}  
This completes the proof. We also illustrate the generating procedure in Figure~\ref{LieA}.
\end{proof}

\begin{figure}[h]
\begin{center}
\includegraphics[scale=.42]{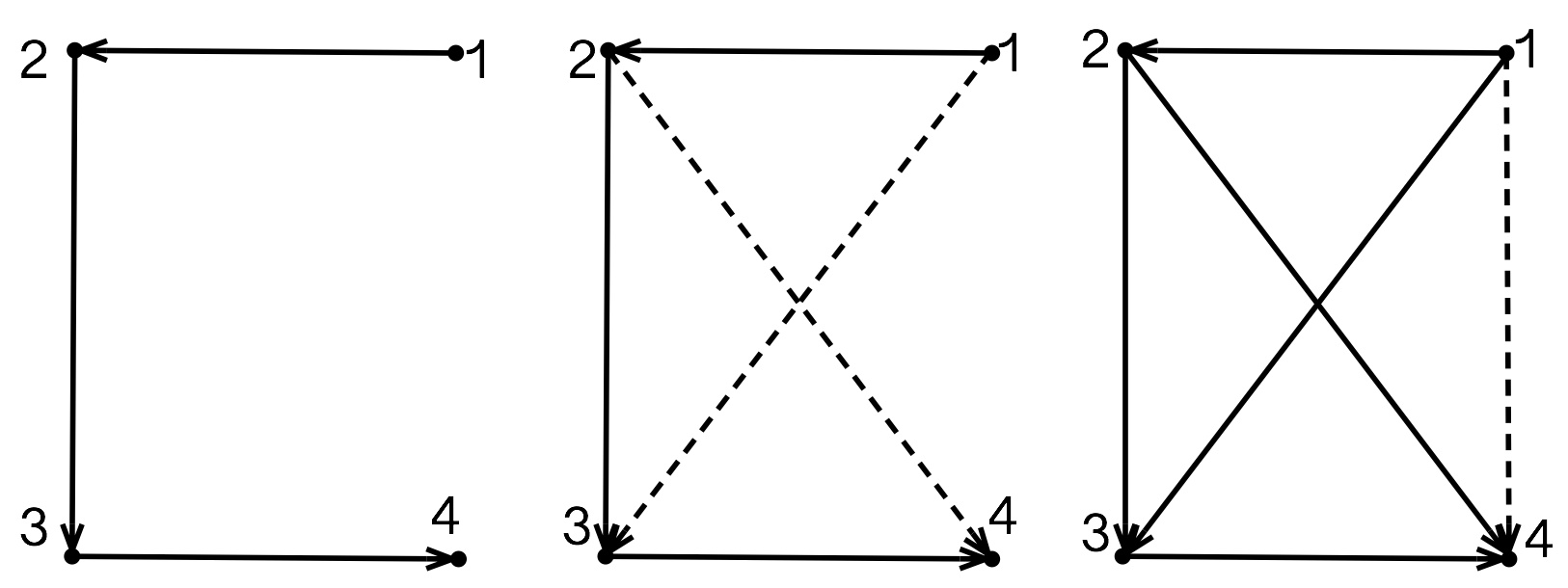}
\caption{We illustrate here the generating process of Lie brackets of $A_{ij}$'s. In this example, we start from $A_{12},A_{23},A_{34}$. Taking Lie brackets, we get $A_{13} = [A_{12},A_{23}] + A_{12}$ and  $A_{24} = [A_{23},A_{34}] + A_{23}$.  Taking the Lie bracket of the original matrix and the newly obtained ones, we get $A_{14} = [A_{13},A_{34}] + A_{13} = [A_{12},A_{24}] + A_{12}$.}
\label{LieA}
\end{center}
\end{figure}

From Corollary~\ref{COR}, it should be clear now that  $\ol{\bb{A}}_G$ contains the vectors space $ \mathbb{A}_{\overline{G}}$. Hence, Proposition~\ref{LIEAL} will be established  after we prove  the following Lemma.

\begin{lem}\label{CLOALG}
The vector space  $ \mathbb{A}_{\overline{G}}$ is closed under the Lie bracket.
\end{lem}

\begin{proof}
We know that the set $A_{\overline{G}}$ is a basis of the vector space  $\mathbb{A}_{\overline{G}}$. So it suffices to show that for any two matrices $A_{ij}$ and $A_{i'j'}$ in $A_{\overline{G}}$, the Lie bracket $[A_{ij}, A_{i'j'}]$ is a linear combination of matrices in $A_{\overline{G}}$. This is a direct consequence of  Lemma~\ref{COMP}; indeed, in case $i'\neq i$ and $i'\neq j$, we have 
\begin{equation*}
[A_{ij}, A_{i'j'}] = 0,
\end{equation*}
 in case $i' = i$, then 
\begin{equation*} 
[A_{ij}, A_{ij'}] = A_{ij} - A_{ij'},
\end{equation*}
and in case  $i' = j$, then 
\begin{equation*}
[A_{ij}, A_{jj'}] = A_{ij'} - A_{ij}
\end{equation*}
with  $A_{ij'}$ in $A_{\overline{G}}$ because $v_i\to v_j \to v_{j'}$ is a path in $G$ which implies that  $v_i\to v_{j'}$ is an edge of $\overline{G}$.
\end{proof}

Lemma~\ref{CLOALG}  implies that the vector space $\mathbb{A}_{\overline{G}}$ contains the Lie algebra closure  of $\bb{A}_{G}$. Proposition~\ref{LIEAL} then follows by combining Corollary~\ref{COR} and Lemma~\ref{CLOALG}.

\section{Lie Algebra of Control Vector Fields}\label{LIEF} 

We now prove the controllability of system~\eqref{MODEL} by verifying the Lie algebra rank condition over the path-connected, open dense set $Q$ defined in~\eqref{PCOD}.
We first rewrite system equation~\eqref{MODEL} into a matrix form which makes it simpler to evaluate the Lie brackets of the control vector fields. To this end, we re-order the entries of the vector $p$ as follows.  Let $x^j_i$ be the $j$-th coordinate of agent $i$, and let
\begin{equation*}
x^j := \left (x^j_1,\ldots,x^j_N\right )
\end{equation*}
be a vector in $\mathbb{R}^N$ collecting the $j$-th coordinate of all agents. In the remainder of this section, the configuration vector  $p$ is taken as 
\begin{equation*}
p = \left (x^1,\ldots,x^n\right ).
\end{equation*}
Let $A$ be an $N$-by-$N$  matrix, and let
\begin{equation*}
D (A) := Diag(A,\ldots,A) 
\end{equation*}  
 be a block-diagonal matrix with $A$ repeated $n$ times. Then, with the notations above, system~\eqref{MODEL} can be expressed  as
 \begin{equation*}
 \dot p = \sum_{i\to j\in E} u_{ij} D(A_{ij}) p 
 \end{equation*}
with $u_{ij}$'s the scalar controls. Note that the control system above  is in a standard affine control form~\cite{bloch2003,SUS} with 
\begin{equation*}
g_{ij}(p) := D( A_{ij}) p, \hspace{10pt} i\to j\in E 
\end{equation*}
the control vector fields.

\begin{Definition}[Lie algebra rank condition]
Let $\mathcal{L}$ be the Lie algebra generated by the control vector fields $g_{ij}$'s.  Let $\mathcal{L}_p$ be the vector space obtained by evaluating the elements of $\mathcal{L}$ at $p$.  We say that $\cal{L}_p$ satisfies the {\bf Lie algebra rank condition} if 
$$
\dim \mathcal{L}_p =\dim P =  n\times N.
$$
\end{Definition}

We now establish the following result.

\begin{pro}\label{CONTROL}
Let $Q$ be the path-connected, open dense subset of $P$ defined by~\eqref{PCOD}. Then, under the assumption of Theorem~\ref{MTHM},  $\cal{L}_p$ satisfies the Lie algebra rank condition for all $p\in Q$.
\end{pro} 

We prove Proposition~\ref{CONTROL} below.  
First, note that for any two matrices $A_{ij}$ and $A_{i'j'}$ in $  \mathbb{A}_{\overline{G}}$, we have 
\begin{equation*}
[D(A_{ij}), D(A_{i'j'}) ] = D\left( [A_{ij}, A_{i'j'}] \right ) .
\end{equation*}
Thus, by Proposition~\ref{LIEAL}, we have 
\begin{equation}\label{LVF}
\mathcal{L}_p = \left \{D(A) p \mid A\in  \mathbb{A}_{\overline{G}} \right \}
\end{equation}
It suffices to show that there are $(n\times N)$ linearly independent vectors in $\mathcal{L}_p$.

 We start the proof with a special case. Consider a strongly connected digraph $G$ with $N$ vertices, with $N = n + 1$. From Corollary~\ref{cor:threecases}, the set $Q$ is an open dense subset of $P$. We now establish the following result.

\begin{lem}\label{SCAS}
Let $G$ be a strongly connected digraph with $(n+1)$ vertices. Then, for any $p \in Q$,  we have $$\dim \mathcal{L}_p = n(n+1).$$
\end{lem}

\begin{proof}
Since $G$ is strongly connected, the transitive closure $\overline{G}$ is a complete graph. Hence there are $n(n+1)$ matrices in $A_{\overline{G}}$. We thus  need to show that the vectors  $\left \{D(A_{ij}) p\mid i\to j\in {\overline{G}} \right \}$  are linearly independent. This is equivalent to showing that if $D(A) p = 0$ for some $A\in A_{\overline{G}}$, then $A = 0$.  To do so,   we first introduce a  matrix as follows: 
\begin{equation*}
X_e := ({\bf 1},x^1,\ldots,x^n) \in \R^{(n+1)\times (n+1)}
\end{equation*}
where   $\mathbf{1}$ is the  vector of all ones in $\mathbb{R}^{n+1}$. We show that the matrix $X_e$ is nonsingular if $p\in Q$; indeed,  consider the following elementary row operation on $X_e$: let
\begin{equation*}
\widetilde X_e := RX_e
\end{equation*}
with $R$ given by
\begin{equation*}
R : = 
\begin{pmatrix}
1 & 0 & \ldots & 0\\
-1 & 1 & \ldots & 0\\
\vdots & \vdots & \ddots & \vdots\\
-1 & 0 & \ldots & 1 
\end{pmatrix}.
\end{equation*}
Then, by  computation, we have
\begin{equation*}
\widetilde X_ e = 
\begin{pmatrix}
1 & x^{\top}_1 \\
0 & (x_2 - x_1)^{\top}\\
\vdots  & \vdots \\
0 & (x_{n+1} - x_1)^{\top}
\end{pmatrix}.
\end{equation*}
Since $p\in Q$, $p$ is non-degenerate. Hence, the $n$ vectors $\{x_2 - x_1,\ldots,x_{n+1} - x_1\}$ are linearly independent. This shows that $\widetilde X_e$ is nonsingular,  and so is $X_e$.  On the other hand, if $D(A)  p = 0$, then $Ax^i = 0$ for all $i = 1,\ldots, n$.  Furthermore, $A\mathbf{1} = 0 $ since $A$ is a zero-row sum matrix. Thus, we have $A X_e = 0$. Since $X_e$ is nonsingular, we have $A = 0$.  This completes the proof. 
\end{proof}

We now  prove Proposition~\ref{CONTROL}.

\begin{proof}[ Proof of Proposition~\ref{CONTROL}] We prove the result by directly constructing a set of $(n\times N)$ linearly independent vectors in $\mathcal{L}_p$.  Suppose that $G_1, \ldots, G_q$ form the coarse strong component decomposition of $G$. Let $N_i$ be the number of vertices of $G_i$. Without loss of generality, we label the vertices of $G$ so that the first $N_1$ vertices are of $G_1$, the next $N_2$ vertices are of $G_2$, and so on so forth. Recall that $W_+$, defined in section~\ref{sec:prelimdig}, is the maximal set of the skeleton of $G$. and from the assumption of Theorem~\ref{MTHM},  we have $N_i > (n+1)$ for all $w_i\in W_+$. We first prove the result for the case when $W_+$  is a singleton, and then show how to lift this assumption.

Without loss of generality, we assume that $W_+ = \{w_1\}$. Let $p_1$ be the sub-configuration of $p$ associated with $G_1$. By assumption, $p$ is contained in $Q$. So then, the sub-configuration $p_1$ is non-degenerate in $\mathbb{R}^n$, and hence,  there must be  $(n+1)$ vectors, say $x_1,\ldots, x_{n+1}$,  such that  $p':=(x_1,\ldots, x_{n+1})$ is a non-degenerate configuration in $\mathbb{R}^n$. Now define 
\begin{equation*}
L_{p'}:= \{D( A_{ij}) p\mid 1\le i,j\le n+1, i\neq j\}.
\end{equation*}
From Lemma~\ref{lem:transclos}, $\ol{G}_1$ is a complete graph. Thus,  $v_i\to v_j$ is an edge of $\ol{G}_1$ (and hence, of $\ol{G}$) for all $1\le i,j\le n$. Thus, by Proposition~\ref{LIEAL}, we have $L_{p'} \subseteq \cal{L}_p$. There are $n(n+1)$ vectors in $L_{p'}$, and from Lemma~\ref{SCAS}, the vectors in $L_{p'}$ are linearly independent.

The remaining $n\times (N-n-1)$ vectors are constructed as follows. Since the configuration $p' = (x_1,\cdots, x_{n+1})$ is non-degenerate in $\bb{R}^n$, for each  $x_i$ with $i \ge (n+2)$, we know from Proposition~\ref{NUMN} that  there are $n$ vectors $\{x_{i_1},\ldots,x_{i_n}\}$ out of $\{x_1,\cdots, x_{n+1}\}$ such that these $n$ vectors, together with $x_i$ form a non-degenerate configuration in $\mathbb{R}^n$.  Now, with the choice of the $n$ vectors $\{x_{i_1},\ldots,x_{i_n}\}$, we define
\begin{equation*}
L_{x_i} := \{D( A_{ij}) p \mid j=i_1,\ldots,i_n \}. 
\end{equation*}
Since $w_1$ is the maximal element, for any vertex $w_k$ of the skeleton digraph $H$, there is a path from $w_k$ to $w_1$. Using the fact that $G_i$'s are strongly connected, we know that for any vertex $v_j$ of   
$G_1$,  there is a path from $v_i$ to $v_j$ in $G$, and hence, $v_i\to v_j$ is an edge of $\ol{G}$. Then, using Proposition~\ref{LIEAL} again, we know that $L_{x_i}\subseteq \cal{L}_{p'}$.   

We now show that the vectors in $L_{x_i}$ are linearly independent. To see this, we define
\begin{equation*}
X:= (x^1,\ldots,x^n)\in\bb{R}^{N\times n}.
\end{equation*}
Note that the vector $D(A_{ij})p$ is derived by  concatenating the $n$ vectors $A_{ij}x^1,\ldots, A_{ij}x^n\in \R^N$. Thus, the $n$ vectors in $L_{x_i}$ are linearly independent if and only if the $n$ matrices $\{A_{ii_1}X,\ldots, A_{ii_n}X\}$ are linearly independent.  By computation, the matrix $A_{ij}X$ satisfies the following condition: the $i$-th row of $A_{ij}X$ is $(x_j-x_i)^{\top}$ while all the other entries are zeros. 
Thus, it suffices to show that the $n$ vectors $x_{i_1}-x_i,\ldots,x_{i_n}-x_i$ are linearly independent. 
But, this follows from the fact that the configuration formed by the $(n+1)$ agents $x_{i_1},\ldots,x_{i_n},x_i$ is non-degenerate in $\R^n$. We have thus proved that the $n$ vectors in $L_{x_i}$ are linearly independent.

The computation above furthermore shows the following fact: choose another $i' = (n+2),\ldots, N$, and  choose $x_{i'_1},\ldots,  x_{i'_n}$ such that these $n$ vectors, together with $x_{i'}$ form another non-degenerate configuration in $\R^n$. Construct $L_{x_{i'}}$ in the same way as  $L_{x_i}$. Then, vectors in $L_{x_{i'}}$ are linearly independent of the vectors in $L_{x_i}$. 
Indeed, if $i\neq i'$, then the positions of the nonzero rows of $A_{ij}X$ and $A_{i'j'}X$ are different for any $j$ and $j'$, and hence, 
\begin{equation*}
\tr\left ((A_{i'j'} X)^\top A_{ij} X\right )  =  (D(A_{ij}) p)^\top D(A_{i'j'} )p = 0
\end{equation*} 
where $\tr(\cdot)$ is the trace of a matrix.  
Using the same argument, we can show  that the vectors in $L_{x_i}$ are  linearly independent of vectors in $L_{p'}$. Now define 
\begin{equation*}
L:= L_{p'}\cup L_{x_{n+2}} \cup\ldots\cup L_{x_{N}}.
\end{equation*}
Then, by construction, there are $n(n+1)$ vectors in $L_{p'}$,  and $n$ vectors in each $L_{x_i} $ for $i=(n+2),\ldots,N$. So then, there are  
$$
 n(n+1) + \sum^N_{i= n+2} n =  n\times N
$$ 
vectors in $L\subset\mathcal{L}_p$, and they are linearly independent. Thus, we have  
$$
\dim \mathcal{L}_p =\dim P =  n\times N, 
$$ 
and hence, $\cal{L}_p$ satisfies the Lie algebra rank condition.

To conclude, we point out  that the same analysis can be applied to the general case  $|W_+|>1$. Without loss of generality, we assume that  $W_+ = \{1,\ldots, k\}$ for some $k < q$. 
Let $p_1,\ldots,p_k$ be the sub-configurations of $p$ associated with the subgraphs $G_1,\ldots, G_k$ of $G$, respectively.  By assumption, $p$ is contained in $Q$. So then, $p_1,\ldots, p_k$ are non-degenerate configurations in $\R^n$. Hence, for each $i=1,\ldots, k$, there is a non-degenerate sub-configuration $p'_i $ of $p_i$ in $\R^n$ which is comprised of $(n+1)$ agents. Thus, we can construct $L_{p'_1},\ldots, L_{p'_k}$ in the same way as $L_{p'}$ in the previous case. Each $L_{p'_i}$ contains $n(n+1)$ linearly independent vectors. Furthermore, if $i\neq j$, then the positions of nonzero entires of vectors in $L_{p'_i}$ are different from those of vectors in $L_{p'_j}$. Thus, vectors in $L_{p'_i}$ are perpendicular to (and hence, independent of) vectors in $L_{p'_j}$.   

Now consider an agent $x_j$ which is not contained in $p'_i$ for any $i = 1,\ldots, k$. From the definition of $W_+$ and the fact that the $G_i$'s are strongly connected, we know that 
there exists at least a subgraph $G_{i}$, for $i= 1,\ldots, k$, such that for any vertex $v_{i_l}$ of $G_i$, there exists a path from $v_j$ to $v_{i_l}$. So then, $v_{j}\to v_{i_l}$ is an edge of $\ol{G}$ for any vertex $v_{i_l}$ of $G_i$. Thus, we can construct $L_{x_j}$ as in the case when $W_+ = \{1\}$, but replace $p'$ with $p'_i$. 
To the end, there are $n$ vectors in $L_{x_j}$, and the vectors in $L_{x_j}$ are perpendicular to vectors in $L_{x_{j'}}$ as long as $j\neq j'$. Let $p' $ be the sub-configuration of $p$ derived by taking the union of $p'_i$, for $i = 1,\ldots, k$. Define
$$
L := \left (\bigcup^k_{i=1} L_{p'_i} \right )\cup \left (\bigcup_{x_j\notin p'} L_{x_j}  \right).
$$
Then, there are $n\times N$ vectors in $L$, and they are linearly independent. This completes the proof.
\end{proof}

Theorem~\ref{MTHM} is then a consequence of Propositions~\ref{LPCOD} and~\ref{CONTROL}, and the Rachevsky-Chow's Theorem. The path controllability is a consequence of a result of Sussmann and Liu~\cite{SUS}. 

\begin{Remark}
We point out that the condition  $p\in Q$ is also a \emph{necessary} condition for $\cal{L}_{p}$ to satisfy the Lie algebra rank condition, and thus a necessary condition for controllability. This follows from the fact that the operations of taking the Lie bracket of the control vector fields  and of taking rotations  of a configuration $p$ commute. Now, if the configuration $p$ is degenerate of rank $k$, we can always rotate it by $\Theta$ so that the last $n-k$ coordinates of each agent in $\Theta\cdot p$ are zero and thus the last $n-k$ entries of the corresponding vector fields for each agent  are zero. Taking the Lie bracket of such vector fields always results in a vector field with the last $n-k$ coordinates for each agent  being zero. They thus form an involutive Lie algebra of dimension $kN$ and thus do not pass the Lie algebra rank condition.


\end{Remark}

\section{Conclusions}
In this paper, we have investigated the controllability of a bilinear formation control model with underlying network topology described by a directed graph $G$. We have shown that the system is approximately path-controllable over the path-connected, open dense subset $Q$ (defined in~\eqref{PCOD}) provided that  $G$ is weakly connected and each maximal component of $G$ has more than $(n+1)$ vertices. To establish the result, we have exhibited some  relations between the transitive closure of $G$ and  the Lie algebra closure of a set of zero row-sum matrices that arose naturally in the study of the formation control model.  Future work may focus on designing explicit control laws for steering the system to follow a specific path, and computing the least number of $u_{ij}$'s  that are in need for the controllability of system~\eqref{MODEL}. 


\bibliographystyle{unsrt}
\bibliography{FC}

\section*{Appendix}

We prove here Lemma~\ref{pksm} stated in Section~\ref{sec:nondeg}, and Lemma~\ref{lem:transclos} stated in Section~\ref{sec:LieAlg1}
\begin{proof}[Proof of Lemma~\ref{pksm}]
Letting $p\in P^k$, we show that there is an open neighborhood $U\subset \R^{nN}$ of $p$, an open neighborhood $V\subset \R^{nN}$ of $0$, and a diffeomorphism $f: U\to V $ such that 
\begin{equation*}
f(U\cap P^k) = V \cap \R^{d_k} 
\end{equation*}
where $\R^{d_k}$ is a linear subspace of $\R^{nN}$ with the last $(nN - d_k)$ entries being zeros. 
For simplicity, but without loss of generality, we assume that 
\begin{equation*}
\operatorname{rank} \{x_2 - x_1,\ldots, x_{k+1} -  x_1\} = k.
\end{equation*}  
Denote by $p_1$ the sub-configuration formed by $x_1,\ldots, x_{k+1}$, and by $p_{-1}$ the sub-configuration formed by the remaining agents. 
 
Choose an open neighborhood $U_1\subset \R^{n(k+1)}$ of $p_1$ such that if $p'_1 = (x'_1,\ldots, x'_{k+1})\in U_1$, then 
\begin{equation*}
\operatorname{rank} \left\{x'_2 - x'_1,\ldots, x'_{k+1} -  x'_1\right\} = k.
\end{equation*}
Choose  any open neighborhood $U_{-1}\subset \R^{n(N-k-1)}$  of $p_{-1}$, and let
\begin{equation*}
U := U_1\times U_{-1}.
\end{equation*}
Then, $U$ is an open neighborhood of $p$ in $\R^{nN}$. For each $p'\in U$, define an $n$-by-$k$ matrix as follows:
\begin{equation*}
A_{p'_1} := (x'_2 - x'_1,\ldots, x'_{k+1} -  x'_1)\in \R^{n\times k}. 
\end{equation*}
Choose an $n\times (n-k)$ matrix $B_{p'_1}$  such that $B_{p'_1}$ is of full-column rank, i.e., 
$
\operatorname{rank} B_{p'_1} = n-k 
$,  
and moreover, the columns of $B_{p'_1}$ are perpendicular to columns of $A_{p'_1}$, i.e.,  
$
B^\top_{p'_1} A_{p'_1} = 0
$. 
Furthermore, we can choose $B_{p'_1}$ so that it depends smoothly on $p'_1\in U_1$; indeed, we first find an $n\times (n - k)$ matrix $B$ so that by shrinking $U_1$ if necessary, the matrix $(A_{p'_1}, B)$ is nonsingular for all $p'_1\in U_1$. Then, by applying the Gram-Schmidt process, we get $B_{p'_1}$. 
Now for each $p'_1\in U_1$, define an $n$-by-$n$ matrix as follows: 
\begin{equation*}
L_{p'_1} := \left (A_{p'_1},B_{p'_1} \right )^\top \in \R^{n\times n}.
\end{equation*}
Then, by construction, $L_{p'_1}$ is invertible, depending smoothly on $p'_1\in U_1$.  

We now construct the diffeomorphism $f: U \to \R^{nN}$.  First, for each $p' = (x'_1,\ldots, x'_N)\in U$,  we define   
 a set of vectors $v_1(p'),\ldots,v_N(p')$ in $\R^n$ as follows:
\begin{equation}\label{eq:defvip}
v_i(p') :=
\left\{
\begin{array}{ll}
x'_i - x_i & \text{if } i=1,\ldots, k+1,\\
L_{p'_1} x'_i - L_{p_1} x_i & \text{if } i = k+2,\ldots, N.
\end{array}
\right.
\end{equation}
Then, we define 
\begin{equation*}
f: p'\mapsto \left (v_1(p'),\ldots, v_N(p') \right).
\end{equation*}
From~\eqref{eq:defvip}, we know that the map $f$ is smooth and open. Let $V$ be the image of $f$, i.e., $V: = f(U)$. 
Since $f(p) = 0$, we have that $V$ 
is an open neighborhood of $0$ in $\R^{nN}$. 

We now show that the map $f: U \to V$ is 
invertible. Pick a vector $(v_1,\ldots, v_N)\in V$ with $v_i\in\R^n$. Let
\begin{equation*}
p'_1 := (v_1+x_1,\ldots,v_{k+1} + x_{k+1} ). 
\end{equation*} 
Then, $p'_1\in U_1$, 
and hence, $L_{p'_1}$ is invertible. Define
\begin{equation*}
p'_{-1} := L^{-1}_{p'_1} \left (v_{k+2} + L_{p_1}x_{k+2}, \ldots, v_{N} + L_{p_1}x_{N} \right),
\end{equation*}
Then, the map $f^{-1}: V\to U$ defined by
\begin{equation*}
f^{-1}: (v_1,\ldots, v_N)\mapsto (p'_{1}, p'_{-1})
\end{equation*}
is the inverse of $f$, and it is also smooth by the construction. Thus, $f$ is a diffeomorphism between $U$ and $V$. 

We now show that the image of $U \cap P^k$ under $f$ is $V \cap \R^{d_k}$. First, note that if $p'\in U$, then $r_{p'} \ge k$ because 
\begin{equation*}
\operatorname{rank} \left \{x'_2 - x'_1,\ldots, x'_{k+1} -  x'_1\right \} = k.
\end{equation*}
Moreover, the equality $r_{p'} = k$ holds if and only if each $x'_i$, for $i=k+2,\ldots, N$, is in the column space of $A_{p'_1}$. Equivalently, $r_{p'} = k$ if and only if 
\begin{equation}\label{eq:bp1x10}
B_{p'_1} x'_i = 0, \hspace{10pt} \forall i=k+2,\ldots, N.
\end{equation}
Then, following~\eqref{eq:defvip} and~\eqref{eq:bp1x10}, we conclude that  $r_{p'} = k$ if and only if   the last $(n-k)$ entries of $v_i(p')$ are zero for all $i = k+2,\ldots, N$. The total number of these zero entries are $
(n-k)(N - k - 1)
$. 
On the other hand, all the other entries of $v_i(p')$'s are free to choose as long as $(v_1(p'),\ldots,v_N(p'))$ is in $V$. The number of these free entries is then
\begin{equation*}
nN - (n-k)(N - k - 1)
\end{equation*}
which is equal to $d_k$. 
Thus, we conclude that
\begin{equation*}
f(U\cap P^k) = V\cap \R^{d_k}.
\end{equation*}
This completes the proof.  
\end{proof}  
 
We now prove Lemma~\ref{lem:transclos}. 
 
 \begin{proof}[Proof of Lemma~\ref{lem:transclos}]
 First, note that if $G$ is strongly connected, then its transitive closure $\overline{G}$ is a complete graph. Indeed, for any pair of vertices $(v_i,v_j)$, there is a path from $v_i$ to $v_j$, and also a path from $v_j$ to $v_i$ in $G$.  Thus, each $\ol{G}_i$, for $i=1,\ldots, q$, is a complete graph.  
 
 Next, we show that  that $\overline{G}_1, \ldots, \overline{G}_q$ form the coarse strong component decomposition of $\overline{G}$. Let $V_i$ be the vertex set of $G_i$, and of $\overline{G}_i$ as well. 
 Choose any subset $\{i_1,\ldots, i_m\}$ of $\{1,\ldots, q\}$. Let    
$$V' := \sqcup^m_{j = 1}V_{i_j}, $$ 
and let ${G}'$ be the subgraph of $\overline{G}$ induced by $V'$. It suffices to show that ${G}'$ is not strongly connected since otherwise, we would have a strong component decomposition with the number of components strictly less than $q$.  The proof is done by contradiction. Suppose that $G'$ is strongly connected; then for any vertex $v_{i'_k}\in V_{i_k}$ and  any vertex $v_{i'_l}\in V_{i_l}$, for $k\neq l$, there is a path from $v_{i'_k}$ to $v_{i'_l}$ in $G'$ (and hence, in $\ol{G}$). So then, by the definition of transitive closure (Def.~\ref{Def: Tranclos}), there is a path from $v_{i'_k}$ to $v_{i'_l}$ in $G$.  Thus, we have that $w_{i_l} \succ w_{i_k}$. But conversely, we can apply the same argument, and have that $w_{i_k}\succ w_{i_l}$ which is a contradiction. So then, the subgraph $G'$ can not be strongly connected.  
Thus, we have shown that $\overline{G}_1, \ldots, \overline{G}_q$ form the coarse strong component decomposition of $\ol{G}$.
 
 It now remains to show that $\ol{H}$ is the skeleton digraph of $\ol{G}$. First, we show that $\ol{H}$ is acyclic. Suppose not, then there is a cycle 
 $$
 w_{i_1}\to w_{i_2} \to \ldots \to w_{i_k} \to w_{i_1}
 $$
 contained in $\ol{H}$. Then again, by the definition of transitive closure, there is a path from $w_{i_l}$ to $w_{i_{l+1}}$, for all $l = 1,\ldots, k-1$, and a path from $w_{i_k}$ to $w_{i_1}$.  Using these paths, we obtain a cycle in $H$ which contradicts the fact that $H$ is acyclic. Thus, the digraph $\ol{H}$ is acyclic.  It now suffices to show that  if there is a path 
 from $w_i$ to $w_j$ in $H$, then $w_i\to w_j$ is an edge of the digraph $\ol{H}$. Let 
 $$
 w_{i_1}\to w_{i_2} \to \ldots \to w_{i_k} 
 $$
 be the path from $w_i$ to $w_j$ with ${i_1} = i$ and ${i_k} = {j}$. Pick a vertex $v_{i_l}\in V_{i_l}$ for each $l = 1,\ldots, k$. Since $w_{i_{l}}\to w_{i_{l+1}}$ is an edge of $H$, from the definition of the skeleton digraph, there is a path from a vertex of $G_{i_l}$ to a vertex of $G_{i_{l+1}}$. But since $G_{i_l}$ and $G_{i_{l+1}}$ are strongly connected, there is a path from $v_{i_l}$ to $v_{i_{l+1}}$ in $G$, and this holds for each $l = 1,\ldots, k-1$. Using these paths, we then obtain a path from $v_{i_1}$ to $v_{i_k}$ in $G$, and hence, $v_{i_1}\to v_{i_k}$ is an edge of $\ol{G}$. Since $v_{i_1}\in V_i$  and $v_{i_k}\in V_j$,  we conclude that $w_i\to w_j$ is an edge of $\ol{H}$.   
 This completes the proof.
 \end{proof}

\end{document}